\documentclass[journal, twocolumn]{IEEEtran}
\usepackage{cite}
\usepackage{amsmath,amssymb,amsfonts,amsthm,bm}
\usepackage{algorithm}
\usepackage{algorithmic}
\usepackage{graphicx}
\usepackage{textcomp}
\usepackage[table,xcdraw]{xcolor}
\usepackage{subfigure}
\usepackage{caption}
\usepackage{multirow} 
\usepackage{colortbl}
\usepackage{balance}

\usepackage{fancyhdr}
\def\BibTeX{{\rm B\kern-.05em{\sc i\kern-.025em b}\kern-.08em
		T\kern-.1667em\lower.7ex\hbox{E}\kern-.125emX}}

\newtheorem{lemma}{Lemma}

\begin{document}
	\title{A Rotation-based Method for  Precoding in Gaussian 
		MIMOME Channels\\
	}
	
	\author{\IEEEauthorblockN{Xinliang Zhang, Yue Qi, 
	\IEEEmembership{Student Member, IEEE}, and 
	Mojtaba Vaezi, \IEEEmembership{Senior Member, IEEE}}
	\thanks{This paper was partially presented in the IEEE International 
	Symposium on Personal, Indoor, and Mobile Radio Communications, September 2020 \cite{zhang2020new}.}
\thanks{The authors are with the Department of Electrical and Computer Engineering, Villanova University, Villanova, PA, USA (e-mail:\{xzhang4, yqi, mvaezi\}@villanova.edu).}

	}

	\maketitle

	\begin{abstract}
		The problem of maximizing secrecy rate of  multiple-input 
		multiple-output multiple-eavesdropper (MIMOME) channels 
		with  
		arbitrary numbers of antennas at each node is studied in this paper. 
		{ First, the optimization problem 
		corresponding to the secrecy capacity of the MIMOME  channel is 
		converted to an equivalent optimization   based on Givens rotations 
		 and eigenvalue decomposition of the covariance matrix. In this new 
		 formulation,  precoder is a rotation matrix which results in a positive 
		 semi-definite (PSD) covariance matrix  by construction. This removes the 
		 PSD matrix constraint and makes the problem easier to tackle.} Next, 
		 {a 
		Broyden-Fletcher-Goldfarb-Shanno (BFGS)-based algorithm} is 
		developed to find the rotation and power allocation parameters. 
		{ Further, the generalized singular value decomposition 
		(GSVD)-based precoding is used to initialize this algorithm.}   The 
		proposed {rotation-BFGS} method
		provides an efficient approach to find a near-optimal transmit strategy for 
		the MIMOME channel and outperforms various  
		state-of-the-art analytical and numerical methods. In particular, the rotation-BFGS 
		precoding achieves higher secrecy rates than the celebrated GSVD precoding,  with a reasonably higher computational 
		complexity.  Extensive numerical results  elaborate on the effectiveness of 
		the rotation-BFGS precoding.
		The new framework developed in this paper can be applied to a variety of similar problems in the context of multi-antenna channels with and without secrecy. 
		
	\end{abstract}
	
	%
	\section{Introduction}\label{sec_intro}
	

	As a complement to higher-layer security measures, \textit{physical layer security} has emerged as a significant technique for security in the lowest layer of communication, i.e., the physical layer.  
	Founded on  information-theoretic security,  which is  built on classical 
	Shannon's 
	 notion of perfect secrecy, physical layer security can offer unbreakable 
	security,  unlike conventional secret-key-based cryptosystems.
	Physical layer security was laid   in
	the 1970s in Wyner’s seminal  work on the \textit{wiretap channel} \cite{wyner1975wire}  where the idea of secure communication   based on the 
	communication channel itself without 
	using encryption keys was first 
	introduced. In this work, 
	Wyner proved that in a wiretap channel (a channel in which a
	transmitter  conveys information to a legitimate receiver in the presence of an  eavesdropper)  communication can be both 
	robust to transmission errors (\textit{reliable}) and  confidential 
	(\textit{secure}), to a certain degree, provided that  the legitimate user's  
	channel is better than the eavesdropper's  
	channel\footnote{Later   in
		the 1990s, Maurer proved that  secret key generation through public communication over
		an insecure yet authenticated channel is possible even when a legitimate user has
		a worse channel than an 
		eavesdropper\cite{maurer1993secret}.}. He    
		established the capacity of the \textit{degraded} wiretap channel. Later, 
		Csiszar and Korner \cite{csiszar1978broadcast} generalized this result
	to arbitrary, not necessarily degraded, wiretap channels. 
	
	In the past decades, physical layer security has  been applied to 
	enhance  
	the 	classical wiretap channels (e.g., by including more realistic  
	assumptions) and to study advanced wiretap channels (e.g., quantum 
	communication \cite{cai2004quantum, wu2019security}). 
	Particularly, as
	multiple-input multiple-output (MIMO) networks continue to flourish worldwide, a
	significant effort has been made to study the MIMO wiretap channel  
	which allows for the exploitation of
	space/time/user dimensions of wireless channels for
	secure communications. Specifically,
	secrecy capacity of  Gaussian multiple-input multiple-output 
	multiple-eavesdropper (MIMOME) channels  under 
	an average total power constraint was established independently in  \cite{khisti2010secure,oggier2011secrecy,liu2009note}. The capacity result  is 
	abstracted as an optimization problem over  input covariance matrix. This problem is non-convex   and its optimal solution is known only for limited settings \cite{shafiee2009towards, parada2005secrecy,fakoorian2013full, loyka2016optimal,vaezi2017journal}.   
	
	Among notable sub-optimal solutions that can be applied to  the MIMOME channel is 
	the \textit{generalized singular value decomposition} (GSVD)-based precoding 
	\cite{fakoorian2012optimal}. GSVD-based precoding  decomposes transmitted channel matrices into several parallel subchannels and confidential information is transmitted over subchannels where the legitimate user is stronger than the eavesdropper.
	This method gives a  
	closed-form solution for achievable secrecy rate which is relatively fast and is 
	asymptotically optimal at high signal-to-noise ratios (SNRs). However, its performance 
	is   not good at certain settings, e.g., when the 
	eavesdropper has a single antenna while other nodes have multiple antennas
	\cite{vaezi2017journal}. Another important sub-optimal solution is Li \textit{et al.'s}	
	alternating optimization and water filling  (AOWF) algorithm \cite{li2013transmit} 
	which alternates the original optimization problem to a convex problem and  finds the corresponding  Lagrange multipliers in an iterative manner. 
	AOWF is more computationally expensive than GSVD-based precoding but it can provide a better secrecy rate in some settings.  The performance of this method also varies depending on the number of 
	antennas at different nodes. For example, its performance is not as 
	good as the GSVD-based precoding when  the number of 
	antennas at the eavesdropper is greater than that of the transmitter.
	There are also  other numerical solutions for this  
	optimization problem \cite{loyka2015algorithm, steinwandt2014secrecy, li2009transmitter}. 
	 Specifically, in \cite{loyka2015algorithm}, a barrier method based 
	iterative algorithm  is tailored to obtain  global optimal  with guaranteed 
	convergence. However, it is not straightforward to obtain the barrier parameter  
	and this involves a challenging optimization problem. In addition, this solution is devised based on an upper bound  
	custom-made for the MIMOME channel and its extension to other related 
	problems is not straightforward.

	Recently, based on a trigonometric  parameterization of the 
	covariance matrix, a closed-form solution for optimal  
	precoding and power allocation of  the MIMOME 
	channel with two transmit antennas was obtained  in 
	\cite{vaezi2017journal, vaezi2017isit}.
	 This approach in finding the optimal covariance matrix is completely 
	different from existing linear beamforming methods. It does not require  
	degradedness condition of \cite{fakoorian2013full} and \cite{loyka2012optimal}, and thus provides the optimal solution for both full-rank and 
	rank-deficient cases in one shot. The above beamforming and power allocation schemes are, however, limited to two transmit antenna cases, and the optimal transmit covariance matrix is still open in general.
	Givens rotation  has been previously investigated  for quantizer  
	design in MIMO broadcast channels  \cite{sadrabadi2006channel}.
		Later, Givens rotation is used as a beamformer to maximize  sum-rate  under average  signal-to-noise ratio (SNR) with limited feedback  in   multiuser multiple-input single-output (MISO) channels  \cite{de2010optimized}. 
		However, Givens rotation theory is the first time to be applied on MIMO wiretap channel to the authors' knowledge. 
		
	In this paper, we use Givens rotation matrices to generalize the approach of
		\cite{vaezi2017isit} to 
		 arbitrary numbers of antennas 
		at each node, and introduce a new method
	 for
	precoding and power allocation in the MIMOME channel. In this 
	approach, without loss of generality, the precoding matrix is formed 
	using a rotation matrix. 
	Then, the covariance matrix can be seen as an operator to 
	appropriately	stretch {(by power allocation)} and rotate 
	{(by precoding)} the input symbols to form a transmit signal that 
	best fits the channels of the legitimate user and 
	eavesdropper. The capacity 
	expression is then transformed into {optimizing rotation angles and 
	power allocation parameters.
		This problem is very different from the original problem, which requires exploring new techniques to solve it. One advantage of the new problem formulation is that it converts the matrix covariance constraint (symmetric and positive semi-definiteness (PSD)) into a set of linear constraints and  thus simplifies the optimization problem.} 
	We then provide a {numerical} solution for the new optimization 
	problem and show that it outperforms the existing methods in various antenna 
	settings. The method in \cite{vaezi2017journal} can be seen 
	as using two-dimension Givens rotation for precoding.

	\subsection{Contributions}\label{sec_intro_sub_ctb}
	The main contributions of 
	this paper are listed below:
	\begin{itemize}

		\item We use a rotation modeling method for the 
		parameterization of the covariance matrix. 
		{This parameterization gives a new representation of the 
		capacity expression for the MIMOME channel  for which more efficient 
		solutions can be exploited. Particularly, the PSD constraint  on the covariance 
		matrix is removed because in the proposed method the covariance matrix is 
		PSD by construction. }
		
		\item For $n_t=1$ and $n_t=2$, the proposed scheme reduces to those in  
		\cite{parada2005secrecy} and \cite{vaezi2017journal}, respectively, for which 
		closed-form solutions are known. For $n_t\ge 3$, finding a closed-form solution is still 
		challenging. In such a case, we introduce a 
			Broyden-Fletcher-Goldfarb-Shanno (BFGS)-based method to iteratively solve 
			the rotation and power allocation parameters.  We name this 
			approach    rotation-BFGS method  in which a rectifier is designed to 
			remove the constraints. Numerical results in different antenna 
			settings confirm that the proposed scheme 
		works better than the well-known GSVD and AOWF. 
		Specifically, the 
		proposed approach outperforms GSVD when $n_e < n_t$, and AOWF 
		approach when $n_e\ge n_t$, where $n_e$ is the number of antennas at the 
		eavesdropper. Particularly, the gap between the proposed  and GSVD-based methods 
		is remarkably high when the eavesdropper has a single antenna. 
		
		\item To improve the  computational complexity of the proposed algorithm (by reducing 
		the number of iterations), we develop an algorithm to  exploit  GSVD as 
		an initialization for our rotation-BFGS method. This initialization 
		improves 
		the results and reduces the computational complexity as it reduces the 
		number of iterations in the optimization problem.
	\end{itemize}
	
	\subsection{Other Related Works}\label{sec_intro_sub_rWork}
%
		An interesting aspect of the proposed  approach is its generality and its great potential for  extension to other related problems.  The MIMOME channel  has turned out as a fundamental  tool for the study of physical-layer security in many other related problems throughout the past decade. Many solutions developed for  the MIMOME
	has appeared to be instrumental in designing transmit strategies that 
	maximize the secrecy rate of extensions of this basic channel model to MIMO 
	channels with multiple eavesdroppers \cite{li2013transmit}, secure relaying 
	\cite{yang2013cooperative,huang2011cooperative}, ergodic secrecy rate 
	\cite{li2011ergodic}, finite alphabet signaling 
	\cite{wu2012linear,zeng2012linear,aghdam2018overview}, 
	artificial noise \cite{aghdam2017joint} and cooperative jammer 
	\cite{yang2016limited}, 
	among others. 
		The rotation method is also applicable to any problem that can be cast 
		as an optimization problem over a covariance matrix.
	Therefore, it is worth studying the optimal covariance matrix of the MIMOME 
	channel as a general tool for  physical-layer security in various MIMO settings.

	\subsection{Organizations and Notations}\label{sec_intro_sub_org}
	The  remainder of this paper is organized as follows. Section~\ref{sec_model} 
	describes the system model and  related works. 
	Section~\ref{sec_rotBasic}
	introduces and elaborates on the  {Givens rotation
		and reformulates} the secrecy capacity for the MIMOME channel. Section~\ref{sec_alg}, 
	details  {a rotation-BFGS} based   algorithm to optimize the 
	achievable secrecy rate.  
	In Section~\ref{sec_V_result},  numerous  simulation results are carried out to 
	demonstrate the 
	effectiveness of the proposed method. 
	Section~\ref{sec_conclu} draws the conclusion.
	
		Notations: Bold lowercase letters denote column vectors and bold 
	uppercase letters denote matrices. 	$|x|$ and 	${\rm log}_2(x)$ denote the absolute 
	value and the binary logarithm of the scalar $x$.  
	$ \mathbf{A}(i,j) $ denotes the entry $(i,j)$ of matrix 
	$\mathbf{A}$. Besides, $(\mathbf{A})^T$, ${\rm tr}(\mathbf{A})$,  
	$|\mathbf{A}|$, 
	and 
	${\rm diag}(\mathbf{A})$ are the transpose,
	trace, determinate,  and diagonal   of the matrix $\mathbf{A}$. 
	{$[\mathbf{A}]^+$ replaces the negative elements with zeros.}
	$\mathbf{I}_n$ is a 
	$n\times n$ identity matrix. 	And $E\{\cdot\}$ is the expectation of 
	random variables.


	\begin{figure}[t]
		\centering
		\includegraphics[width=0.4\textwidth]{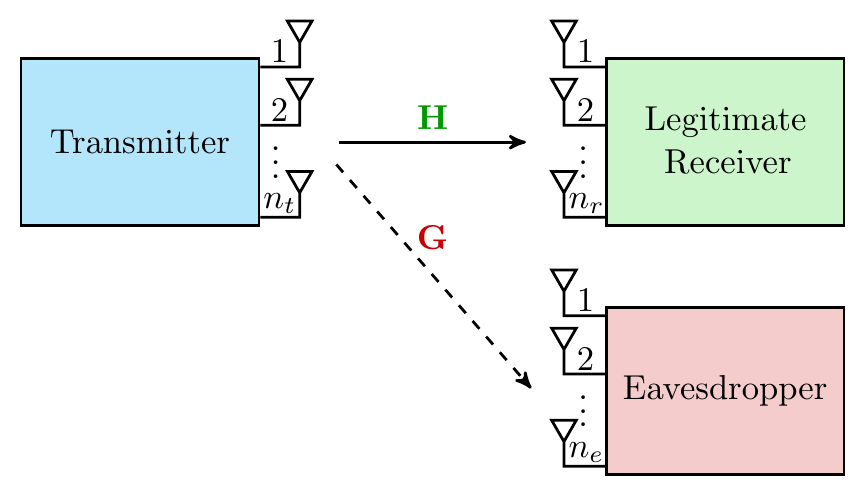}
		\caption{The MIMOME channel with $n_t$, $n_r$, and $n_e$  antennas at the 
		transmitter, legitimate receiver, and eavesdropper.}
		\label{fig_figMIMOME}
	\end{figure}
	
	\section{System Model and Related Works}\label{sec_model}
	\subsection{System Model}\label{sec_model_sub_model}

	We consider a  MIMOME channel  with $n_t$ antennas at the transmitter, 
	$n_r$ antennas at the receiver, and  $n_e$ antennas at the eavesdropper, as depicted 
	in Fig.~\ref{fig_figMIMOME}. The transmitter knows
		the perfect channel state information (CSI) of 
		both users\footnote{A perfect CSI is assumed, as we are deriving the 
		theoretical limits. 
		This may  
		provide an upper bound in terms of achievable secrecy rates. The method we are developing in this paper is, however, applicable to the case with imperfect CSI. In our future 
		works, we will relax this idealized assumption and consider the practical 
		scenarios with imperfect CSI.}.
	The received signals at the legitimate receiver and the eavesdropper
	can be, respectively, expressed as
	\begin{subequations}
		\begin{align}
		\mathbf{y}_r = \mathbf{Hx} + \mathbf{w}_r,\label{eq_recSig1}\\
		\mathbf{y}_e = \mathbf{Gx} + \mathbf{w}_e,\label{eq_recSig2}
		\end{align}
	\end{subequations}
	in which $\mathbf{H} \in \mathbb{R}^{n_r \times n_t}$  and $\mathbf{G} 
	\in \mathbb{R}^{n_e \times n_t}$ are the channels corresponding to the 
	receiver and eavesdropper,  $ \mathbf{x} \in 
	\mathbb{R}^{n_t}$ is the transmitted signal,  and 
	$\mathbf{w_r} \in \mathbb{R}^{n_r}$ 
	and $\mathbf{w_e}\in \mathbb{R}^{n_e}$ are independent and identically 
	distributed (i.i.d)  Gaussian noises with zero means and identity covariance 
	matrices. A representation of secrecy capacity is 
	given by \cite{liu2009note} 
	\begin{align}\label{eq_optOrg}
	\textmd{(P1)}\quad\mathcal{C}_s = \max\limits_{\mathbf{Q}\succeq\mathbf{0}, 
		{\rm tr}(\mathbf{Q})\leq P_t}\frac{1}{2}\log_2 
	\frac{|\mathbf{I}_{n_r}+\mathbf{H}\mathbf{Q}\mathbf{H}^T|}{|\mathbf{I}_{n_e}
		+\mathbf{G}\mathbf{Q}\mathbf{G}^T|}.
	\end{align}
	\noindent Based on Slyvester’s determinant theorem, i.e., 
	$\det(\mathbf{I} + \mathbf{XY})=\det(\mathbf{I} + \mathbf{YX})$, 
	\eqref{eq_optOrg} can also be rewritten as
	\begin{align}\label{eq_optVer1}
	\quad\mathcal{C}_s = \max\limits_{\mathbf{Q}\succeq\mathbf{0}, 
		{\rm tr}(\mathbf{Q})\leq P_t}\frac{1}{2}\log_2 
	\frac{|\mathbf{I}_{n_t}+\mathbf{H}^T\mathbf{H}\mathbf{Q}|}{|\mathbf{I}_{n_t}
		+\mathbf{G}^T\mathbf{G}\mathbf{Q}|},
	\end{align}
   where 
	$\mathbf{Q}=E\{\mathbf{xx}^T\} \in \mathbb{R}^{n_t \times n_t}$ is 
	the covariance matrix of the channel input $\mathbf{x}$ and $ P_t $ 
	is the total transmit power. 
	$ 
	\mathbf{Q} $ is symmetric and  positive semi-definite (PSD)  by definition.
	
	The architecture of the linear precoding and power allocation is depicted in 	Fig.~\ref{fig_figTx}. 
	In this figure, $s_1,\hdots,s_{n_t}$ are input symbols which are independent and identically
	distributed Gaussian random variables with zero means
	and unit variances, $\lambda_{i}$s are power allocation coefficients,    $ \mathbf{V}$ is the precoding matrix, and ${\bf x} = [x_1,\hdots,x_{n_t}]^T$ is the transmit vector whose covariance is  $ \mathbf{Q} $.

	\begin{figure}[t]
		\centering
		\includegraphics[width=0.4\textwidth]{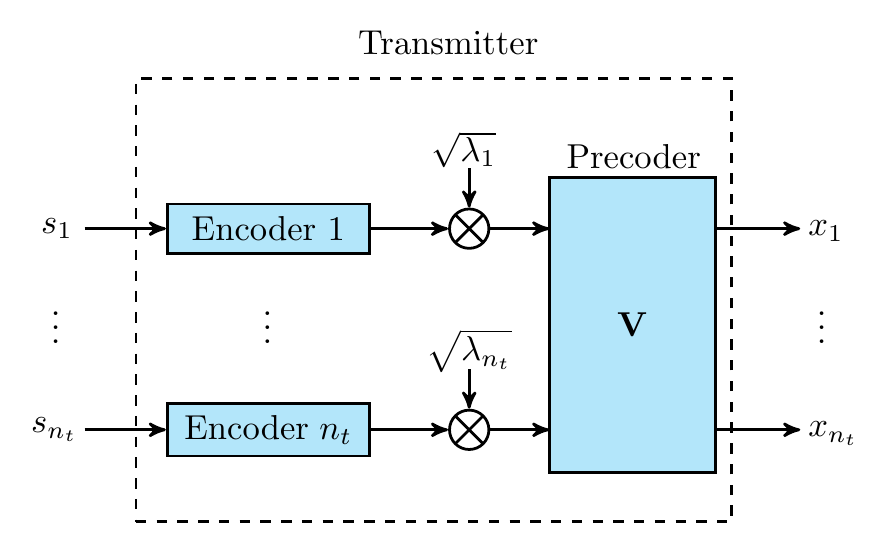}
		\caption{The structure of linear precoding and power 
		allocation.  $s_1,\hdots,s_{n_t}$ are the independent input symbols,  
		$\lambda_1,\hdots,\lambda_{n_t}$ are their corresponding allocated powers, 
		and $x_1,\hdots,x_{n_t}$ are the transmitted signals.}
		\label{fig_figTx}
	\end{figure}
	
	\subsection{Existing Results}\label{sec_model_sub_exist}
	The optimal transmission over the MIMOME channel is still an open problem 
	in general. However, there are a number of notable analytical results for special numbers 
	of antennas as well as numerical results as listed below.
	\subsubsection{Analytical Solutions} \label{sec_model_sub_exist_1}
	An analytical capacity-achieving covariance matrix is known only for special cases. 
	These are limited to:
	\begin{itemize}
		\item 
		$n_t=1$: this is the single-input multiple-output (SIMO) case in which $ 
		\mathbf{Q} $ is a scalar and the optimal solution is either $P_t$ or $0$ 
		\cite{parada2005secrecy}.
		
		\item $n_r=1$: the so-called  multiple-input single-output 
		multiple-eavesdropper (MISOME) channel in which {\it generalized 
			eigenvalue decomposition} of $ \mathbf{H} $ and $ \mathbf{G} $ 
		achieves the capacity \cite{misomeOpt}.
		
		\item $ n_t =2$, $ n_r=2 $, and $n_e=1$: the optimization problem is  shown to be the Rayleigh quotient 	and  optimal signaling, which is the maximum eigenvalue of this problem, 
		is unit-rank 
		\cite{shafiee2009towards}. 
		
		\item  $ n_t =2$: in which the secrecy capacity is obtained  
		by modeling the covariance matrix as a $2\times 2$ rotation matrix 
		\cite{vaezi2017journal, vaezi2017isit}.
		
		\item  $ 
		\mathbf{Q} $ is full-rank (which implies 
		$\mathbf{H}^T\mathbf{H}-\mathbf{G}^T\mathbf{G}\succ\mathbf{0}$) and also $ P_t 
		$ is greater than a certain threshold \cite{fakoorian2013full, loyka2012optimal}: in this case the problem is convex and Karush–Kuhn–Tucker (KKT) conditions are used to find the optimal 	$\mathbf{Q} $.
		
		It is worth noting that, as $n_t$ grows,   few channel realizations satisfy the above conditions. For $n_t=3$, $n_r=3$, and 
		$n_e=1$, for example, the probability of having a full-rank solution is 
		less than $18.2\%$\footnote{This is obtained by Monte Carlo 
		experiments with  
			$10^6$ trails where $\mathbf{H}$ and $\mathbf{G}$  have the same 
			distributions.}. This value  decreases when $n_e$ goes up. Therefore, it 
			can be 
		said that an analytical solution for the MIMOME channel is still an open 
		problem in many practical cases.
	\end{itemize}
	
	\subsubsection{Suboptimal Analytical Solution}\label{sec_model_sub_exist_2}
	For a  general MIMOME channel, a sub-optimal solution  can be obtained 
	using GSVD-based beamforming \cite{fakoorian2012optimal}. By 
	applying GSVD 
	on   $ \mathbf{H} $ and $ \mathbf{G} $, the optimization problem 
	\eqref{eq_optVer1} is simplified to a set of parallel non-interfering channels  whose  optimal power allocation can be obtained using 
	Lagrange multiplier 
	and KKT  conditions. Since parallelization using GSVD 
	does not necessarily convert this problem into an equivalent one,  GSVD-based 
	beamforming is not the optimal solution in general.
		It is not even 
		close enough to the capacity in some cases. For example, 
		GSVD-based 
		beamforming achieves less than $70\%$ of the capacity when $n_t=3$, 
		$n_r=2$, and $n_e=1$.

	\subsubsection{Numerical Solutions}\label{sec_model_sub_exist_3}
	There are still important cases of the MIMOME for which optimal $ 
	\mathbf{Q} $ is unknown. Due to the intractability of the problem in an 
	analytical form, numerical solutions have been developed to tackle this 
	problem. The AOWF
	\cite{li2013transmit},  which is computationally efficient to implement, is one of them. 
	Despite its effectiveness in many cases, AOWF experiences problems when  $n_e$ 
	is greater 
	than  $n_t$, for example, which is caused by a failure in finding an optimal Lagrange multiplier. 
	We modify this issue in this paper, as we will see later in Section~\ref{sec_V_result}. 
	The price is a higher time consumption in the modified approach.

	In the next section, { we apply a rotation-based model} for  the 
	covariance matrix $ 
	\mathbf{Q} $ which is a generalization of the solution in 
	\cite{li2013transmit, vaezi2017journal}, from 
	$n_t=2$ to any arbitrary $n_t$. This model is then used to find transmit signaling that can be used to achieve secrecy capacity of the MIMOME channel regardless of the number of antennas at different nodes.


	\section{A Rotation Modeling of the Problem}\label{sec_rotBasic}
	In the following subsections, we further model matrix $\mathbf{V}$ using the rotation matrix related method by reviewing $n_t=2$  \cite{vaezi2017isit} first and then generalizing it to an arbitrary $n_t$ with proof.
	
	The covariance matrix $\mathbf{Q}$ can be eigendecomposed as
	\begin{align}\label{eq_eig1}
	\mathbf{Q}= \mathbf{V} \mathbf{\Lambda} \mathbf{V}^T,
	\end{align}
	in which $\mathbf{\Lambda}$ is a 
	diagonal matrix and its diagonal elements  
	are the eigenvalues of $ \mathbf{Q} $, which are real and non-negative, i.e.,
	\begin{align}\label{eq_condiPSD}
	\mathbf{\Lambda}={\rm diag}(\lambda_{i}),\;\lambda_{i} \geq 
	0,\;i=1,2,\hdots,n_t,
	\end{align}
	and the total power constraint 
	${\rm tr}(\mathbf{Q})\leq P_t$ will be equivalent to
	\begin{align}\label{eqcondiPt}
	\sum_{i=1}^{n_t}\lambda_i\leq P_t. 
	\end{align}
	Also, $ \mathbf{V} \in \mathbb{R}^{n_t \times 
		n_t}$ is the matrix composed of $n_t$ corresponding eigenvectors of $ 
	\mathbf{Q} $ in $ \mathbb{R}^{n_t} $ vector space. Since $\mathbf{Q}$ is symmetric, the matrix 
	$\mathbf{V}$ is  \textit{orthonormal}.
	
	With this, an immediate implication of the above decomposition is that linear precoding can achieve  the capacity of the MIMOME channel.

	\subsection{Rotation Modeling for $n_t=2$  
	}\label{sec_rotBasic_sub_nt2}
	The capacity region of  MIMOME channel with two transmit antennas has been 
	established in \cite{vaezi2017journal, vaezi2017isit}, using a rotation matrix in 
	two-dimensional (2D) space. In this case, the eigenvalue matrix $\mathbf{\Lambda}$ can 
	be written as
	\begin{align}\label{eq_eigValD2}
	\mathbf{\Lambda}=\left[
	\begin{matrix}
	\lambda_{1}	&	0\\
	0	&	\lambda_{2}
	\end{matrix}
	\right],
	\end{align}
	and, without loss of generality, the 
	eigenvectors  can be written as the following orthonormal (rotation) matrix
	\begin{align}\label{eq_rotD2}
	\mathbf{V}=\mathbf{V}_{12}\triangleq\left[
	\begin{matrix}
	\cos\theta_{12}	&-\sin\theta_{12}\\
	\sin\theta_{12}	&\cos\theta_{12}
	\end{matrix}
	\right].
	\end{align}
	The rotation angle $\theta_{12}$ in rotation matrix $\mathbf{V}_{12}$ 
	corresponds to the rotation from the direction of the 
	standard basis 
	(unit vector) $\mathbf{e}_1=(1,0)^T$ to the standard basis 
	$\mathbf{e}_2=(0,1)^T$ in $\mathbb{R}^2$ vector space. This rotation is 
	achieved on the plane defined by $\mathbf{e}_1$ and $\mathbf{e}_2$. Then, 
	the covariance matrix  $\mathbf{Q}$  can be built using three parameters: 
	the two non-negative eigenvalues $\lambda_{1}$  and  $\lambda_{2}$, and the 
	rotation 
	angle $\theta_{12}$. These parameters are complete to represent any arbitrary 
	$2\times 2$ covariance matrix as in \eqref{eq_eig1}. Then the original 
	optimization problem \eqref{eq_optVer1} can be 
	equivalently converted to
	\begin{subequations}
		\begin{align}
		\quad\mathcal{C}_s = &\max 
		\limits_{\lambda_{1},\lambda_{2},\theta_{12}}\frac{1}{2}\log_2\frac{|\mathbf{I}_{n_t}
			+\mathbf{H}^T\mathbf{H}\mathbf{V}\mathbf{\Lambda}\mathbf{V}^T|}{|\mathbf{I}_{n_t}
			+\mathbf{G}^T\mathbf{G}\mathbf{V}\mathbf{\Lambda}\mathbf{V}^T|},
		\label{eq_optD2_a}\\
		&\;\;\;\operatorname{s.t.} \quad\lambda_{1}\geq0,\; \lambda_{2}\geq0,\; 
		\lambda_{1}+\lambda_{2} \leq P_t.\label{eq_optD2_b}
		\end{align}\end{subequations}
	In light of this modeling, an analytical solution for optimal precoding matrix 
	and power allocation scheme are obtained in 
	\cite{vaezi2017isit} by finding $ \theta_{12} $, $ \lambda_{1} $ and $ \lambda_{2} $. In 
	the next subsection, we extend this method to the cases for 
	an arbitrary $n_t$.

	\subsection{Generalization  to an Arbitrary  
	$n_t$}\label{sec_rotBasic_sub_nt}
	To generalize the rotation modeling method to an arbitrary $n_t\times n_t$ covariance 
	matrix, $\mathbf{\Lambda}\in \mathbb{R}^{n_t\times n_t}$ is a diagonal matrix with 
	non-negative 
	elements $\mathbf{\Lambda}(i,i)\triangleq\lambda_{i}$.
	$\mathbf{V}$ is a 
	rotation matrix in 
	$\mathbb{R}^{n_t\times n_t}$ vector space which can be obtained by
	\begin{align}\label{eq_Vnbyn_}
	\mathbf{V}=\prod_{i=1}^{n_t-1}\prod_{j=i+1}^{n_t} \mathbf{V}_{ij},
	\end{align}  
	in which the basic rotation matrix $ \mathbf{V}_{ij} $ is the Givens matrix 
	\cite{matrix} defined as
	\begin{align}\label{eq_VnD} 
	\mathbf{V}_{ij}=\left[
	\begin{matrix}
	1		&\cdots	&		&			&		&\cdots	&0\\
	\vdots	&\ddots	&		&			&		& &\vdots\\
	&		&v_{ii}	&	\cdots	&v_{ij}	&		&\\
	&		&\vdots	&	\ddots	&\vdots	&		&\\
	&		&v_{ji}	&	\cdots	&v_{jj}	&		&\\
	\vdots		&	&		&			&		&\ddots	&\vdots\\
	0		&\cdots	&		&			&	&\cdots	&1
	\end{matrix}
	\right],
	\end{align}
	and
	\begin{align}\label{eq_VnDsub}
	\left[
	\begin{matrix}
	v_{ii}	&v_{ij}\\
	v_{ji}	&v_{jj}
	\end{matrix}
	\right]
	=\left[
	\begin{matrix}
	\cos\theta_{ij}	&-\sin \theta_{ij}\\
	\sin\theta_{ij}	&\cos \theta_{ij}
	\end{matrix}
	\right].
	\end{align}
	$\mathbf{V}_{ij}$ represents a rotation from the $i$th standard basis 
	to the  
	$j$th standard basis in $ \mathbb{R}^{n_t} $ vector space with a 
	rotation angle  
	$\theta_{ij}$. That is, we show that an arbitrary 
	orthogonal matrix $\mathbf{V}$ can be represented by \eqref{eq_Vnbyn_}. 
	Further, an arbitrary covariance matrix $\mathbf{Q}\in \mathbb{R}^{n_t\times 
		n_t}$ can be represented by $n_t$ non-negative eigenvalues and 
	$\frac{1}{2}n_t(n_t-1)$ rotation angles.  
	
	It should be 
	noted that the order of multiplication in \eqref{eq_Vnbyn_} is not unique, and a different	order will lead to different rotation angles $\theta_{ij}$. In this paper, 
	without loss of generality, we use the order definition in \eqref{eq_Vnbyn_}.
	
	\begin{lemma}\label{lemma_Qn}
		To reach the secrecy capacity of the MIMOME with 
		$n_t\geq 2$, it is sufficient to use a PSD diagonal 
		$\mathbf{\Lambda}$ and the rotation matrix $ \mathbf{V} $ is given in 
		\eqref{eq_Vnbyn_} to generate the input covariance matrix 
		$\mathbf{Q}=\mathbf{V}\mathbf{\Lambda}\mathbf{V}^T$.
	\end{lemma}

	\begin{proof}
		First, we prove that 
		$\mathbf{Q} = \mathbf{V}\mathbf{\Lambda}\mathbf{V}^T$ is a covariance matrix. It is straightforward to check that  $\mathbf{V}$ in \eqref{eq_Vnbyn_} is an 
		orthonormal matrix, i.e., $\mathbf{V}\mathbf{V}^T=\mathbf{I}$.
		Since diagonal elements of $\mathbf{\Lambda}$ are non-negative, $\mathbf{Q}$ is symmetric and PSD, i.e., a 
		covariance matrix.
		
		Next, we prove that an arbitrary covariance matrix $ \mathbf{Q} $ can be written as 
		$\mathbf{Q}=\mathbf{V}\mathbf{\Lambda}\mathbf{V}^T$ while $\mathbf{V}$ is defined as \eqref{eq_Vnbyn_}. It  suffices to  find  $\theta_{ij} $ such that \eqref{eq_theo1_proof1} holds for a given orthonormal  $ \mathbf{V} $
		\begin{align}\label{eq_theo1_proof1}
		\left(\prod_{i=1}^{n_t-1}\prod_{j=i+1}^{n_t} 
		\mathbf{V}_{ij}\right)^T\mathbf{V}=\mathbf{I}.
		\end{align}
		This process is carried out in  Algorithm~\ref{alg_solveAngle}. In fact, 
		\eqref{eq_theo1_proof1} applies a series of Givens rotation on $\mathbf{V}$. 
		They keep the relative orthogonality and rotates $\mathbf{V}$ to identity 
		matrix $\mathbf{I}$.
		
		To better appreciate this, we note that if we expand the product of matrices on the left side of  $\mathbf{V} $ in \eqref{eq_theo1_proof1},
		we see that $\mathbf{V} $ is initially multiplied with $\mathbf{V}_{12}^T $. Then, the corresponding 
		rotation angle $ \theta_{12} $ can be chosen  to set the entry $(2,1)$ of $ 
		\left(\mathbf{V}_{12}\mathbf{V}\right)$ to zero.  Next, $ 
		\mathbf{V}_{13}^T$ is multiplied to the new $\mathbf{V}$ and $ \theta_{13} $ can be chosen  to set the entry $(3,1)$ to zero. 
		This process continues until the last element  under the main diagonal of $\mathbf{V}$, i.e., the entry $(n_t,n_t-1)$,   becomes  zero. We note that, since $ \mathbf{V} $ is orthonormal, the upper triangle also will be zero throughout this process. That is, 
		the left side of \eqref{eq_theo1_proof1} becomes an identity matrix. This proves
		Lemma~\ref{lemma_Qn}.  The rotation 
		angles $\theta_{ij} $ can be obtained by  Algorithm~\ref{alg_solveAngle} which is a 
		generalization of Algorithm 5.1.3 of \cite{matrix} for vectors in $\mathbb{R}^{n_t}$, where ${\rm atan2}(\cdot,\cdot)$ is the 
		four-quadrant inverse tangent  denoted bt \texttt{atan2} in MATLAB.
		
		In certain cases, the eigenvalue decomposition of $\mathbf{Q}$ may give 
		an improper rotation matrix \cite{improper} whose determinate is $-1$ and 
		cannot be converted to an identical matrix by rotation. As a result, 
		$\mathbf{V}_{ij}$ for \eqref{eq_theo1_proof1} does not exist. To deal 
		with this 
		issue, the improper 
		rotation matrix can be converted to a proper rotation matrix by exchanging 
		arbitrary two eigenvectors of $ \mathbf{V} $ and the corresponding 
		eigenvalues of $\mathbf{\Lambda}$.   As an example, we can define 
		$\mathbf{V}^\prime=\mathbf{V}\mathbf{I}^\prime$ and 
		$\mathbf{\Lambda}^\prime=\mathbf{I}^{\prime 
			T}\mathbf{\Lambda}\mathbf{I}^{\prime}$, where
		\begin{align}\label{eq_ReOrder_n}
		\mathbf{I}^\prime=\left[
		\begin{matrix}
		1	& \cdots	&	0	&0&0\\
		\vdots& \ddots& \vdots& \vdots& \vdots\\
		0	& \cdots&	1	&0 & 0\\
		0	& \cdots&	0	&0 & 1\\
		0	& \cdots&	0&1	&0 
		\end{matrix}
		\right]\in\mathbb{R}^{n_t\times n_t}.
		\end{align}
		This specific $\mathbf{I}$  exchanges the last two eigenvectors and eigenvalues of $\mathbf{Q}$. Since 
		$\mathbf{I}^\prime\mathbf{I}^{\prime T}=\mathbf{I}$, it is clear that
		\begin{align}\label{eq:repQ3_n}
		\mathbf{V}^\prime\mathbf{\Lambda}^\prime\mathbf{V}^{\prime T}
		=\left(\mathbf{V}\mathbf{I}^\prime\right)
		\left(\mathbf{I}^{\prime T}\mathbf{\Lambda}\mathbf{I}^{\prime}\right)
		\left(\mathbf{V}\mathbf{I}^\prime\right)^T
		=\mathbf{V}\mathbf{\Lambda}\mathbf{V}^T=\mathbf{Q}.
		\end{align}
		In such a case, $\mathbf{V}^\prime$ is a rotation matrix and the 
		$\mathbf{Q}$ will 
		remain the same. This completes the proof of  Lemma~\ref{lemma_Qn}.
	\end{proof}
	
	\begin{algorithm}[ht]
		\caption{Rotation Angles Solution}\label{alg_solveAngle}
		\begin{algorithmic}[1]
			\STATE	Initialize $[\mathbf{V}, \mathbf{\Lambda}]={\rm 
				eig}(\mathbf{Q})$, $i=1$;
			\IF		{$\det(\mathbf{V})$ = -1}
			\STATE	Exchange first two columns of $\mathbf{V}$;
			\STATE	Exchange first two values on diagonal of $\mathbf{\Lambda}$;
			\ENDIF
			\WHILE	{$i\leq(n_t-1)$}
			\STATE	$j=i+1$;
			\WHILE	{$j\leq n_t$}
			\STATE	$\theta_{ij}=-{\rm atan2}(-\mathbf{V}(j,i),\mathbf{V}(i,i));$
			\STATE	$\mathbf{V}_{\rm rot} = \mathbf{I}_{n_t};$
			\STATE	$\mathbf{V}_{\rm rot}(i,i)=\mathbf{V}_{\rm 
				rot}(j,j)=\cos\theta_{ij}$;
			\STATE	$\mathbf{V}_{\rm rot}(j,i)=-\mathbf{V}_{\rm 
				rot}(i,j)=\sin\theta_{ij}$;
			\STATE	$\mathbf{V}=\mathbf{V}_{\rm rot}\mathbf{V}$;  
			\STATE	$j=j+1$;
			\ENDWHILE
			\STATE	$i=i+1$;
			\ENDWHILE
			\STATE	Output $\theta_{ij}$, $\forall 1\leq i<j\leq n_t$.
		\end{algorithmic}
	\end{algorithm}
	
	{Lemma~\ref{lemma_Qn} shows that any covariance matrix can be 
	formed using a rotation matrix  and a diagonal power allocation matrix. This  is  
	useful in many optimization problems, including that of the MIMOME channel in 
	\eqref{eq_optVer1}, as it removes the PSD constraint of the covariance matrix 
	(i.e., $\mathbf{Q}\succeq\mathbf{0}$). Instead, we will have a set of linear 
	constraints to make sure that the diagonal elements of the power allocation 
	matrix are non-negative and their sum is not greater than $P_t$.} 
	
	{Specifically, for any} $n_t$, the optimization problem 
	\eqref{eq_optVer1} is
	equivalently reformulated as 
	\begin{subequations} \label{eq:P2}
		\begin{align}
		\textmd{(P2)}\quad\mathcal{C}_s = &\max 
		\limits_{{\bm \lambda},{\bm\theta}}\frac{1}{2}\log_2\frac{|\mathbf{I}_{n_t}
			+\mathbf{H}^T\mathbf{H}\mathbf{V}\mathbf{\Lambda}\mathbf{V}^T|}{|\mathbf{I}_{n_t}
			+\mathbf{G}^T\mathbf{G}\mathbf{V}\mathbf{\Lambda}\mathbf{V}^T|},	
		\label{eq_optD3_1}\\
		&\ \operatorname{s.t.} \quad\sum_{i=1}^{n_t}\lambda_{i} \leq P_t, \label{eq_optD3_2}\\
		&\qquad\quad\lambda_{i}\geq0, i\in\{1,\hdots,n_t\},\label{eq_optD3_3}
		\end{align}\end{subequations}
	{  in which we have defined}	\begin{subequations} 
	\label{eq:2vectors}
		\begin{align}
		&{\bm \lambda}\triangleq\{\lambda_{i}\},\;1\leq i \leq n_t,\label{eq_lamVec}\\
		&{\bm \theta}\triangleq\{\theta_{ij}\},\;1\leq i<j\leq n_t\label{eq_theVec}.
		\end{align}
	\end{subequations}
	{ as the compact form of the parameters 
		of \eqref{eq_Vnbyn_}-\eqref{eq_VnDsub}}\footnote{{It is worth 	
		mentioning that Givens rotation method applied 
		in many problems can finally 
			boil down to finding the optimal rotation parameters. However, the application of the rotation method is 
			different. In \cite{de2010optimized}, only the optimal unitary 
			beamforming without power allocation is evaluated because it constrained the precoding matrices to 
			be unitary for limited feedback. On the other hand, the optimization approach is in different manners. In their iterative 
			algorithm, the rotation angles of each covariance matrix are obtained one by one. 
			In our paper, these parameters can be updated vector-wised once.
	}}.
	
	Different from \eqref{eq_optVer1}, this new representation replaces the constraint that $\mathbf{Q}$ is symmetric and PSD 
	by linear constraints on $\bm\lambda$ only while rotation angles can take any number, i.e.,  $\bm\theta\in\mathbb{R}$. Then, numerical methods can be applied to optimize the parameters $\bm\theta$ and
	$\bm\lambda$ to obtain the optimal 
	secrecy rate. For $\mathbf{Q}$ with dimension $n_t$, 
	the required number of eigenvalues is $n_t$, whilst this number is
	$\frac{1}{2}n_t(n_t-1)$ for the rotation angle. The total number of parameters is 
	$\frac{1}{2}n_t(n_t+1)$, which is equal to the number of the elements of the 	upper triangular of $\mathbf{Q}$. Theoretically, the 
	rotation modeling method can provide a systematic approach to traverse 
	$\mathbf{Q}$ by traversing $\bm\lambda$ and $\bm\theta$ in finite regions. 
	The feasible region  of $\bm\lambda$ is given in 
	\eqref{eq_optD3_2}-\eqref{eq_optD3_3} and for each rotation angel the 
	regions can be $[0,2\pi)$\footnote{It can be proved that in some cases the 
		optimal $\theta_{ij}$ is in $[0,\pi) $. See, for example, the case for $n_t=2$ 
		\cite{vaezi2017journal}.}.
	{	The rotation method can be applied to many other problems, 
	some listed in Section~\ref{sec_intro_sub_rWork}.
		%

		{ As mentioned earlier, for $n_t = 2$ a 
			closed-form solution for  (P2) is known in \cite{vaezi2017journal}. However, 
			finding closed-form solutions of $\bm \lambda $ and $ \bm\theta $ is 
			challenging  for  $n_t \ge 3$. Next, we introduce a structure to solve this problem effectively. }
		


		%
		{
			\section{Solving the New Optimization Problem}\label{sec_alg}
			In this section, we develop a novel method to solve (P2).
		Specifically, we convert (P2) to an unconstrained problem and resort 
		to  	BFGS \cite{luenberger1984linear}  to solve it.	We should highlight 
		that various iterative optimization methods, such as gradient descent, 
		Newton's method,   and {quasi-Newton algorithms with 
		interior-point method (that is, MATLAB convex optimization tool {\tt 
		fmincon})}, may be used to 
		optimize  (P2). Among them,
		Newton's method is the fastest, but it requires the local 
		Hessian matrix and its inverse, which  may not exist at some points.  
		Quasi-Newton methods keep the 
		convergence advantage of Newton's method without requiring the local 
		Hessian matrix and its inverse.
			 BFGS algorithm is an outstanding
		variety of quasi-Newton methods that can 
		converge faster for non-convex problems \cite{byrd1989tool, 
			nocedal2006numerical}.
			The proposed method is called rotation-BFGS.  The block diagram of 
			this  method is given in Fig.~\ref{fig_obj_sys} in which the inputs are 
			channel matrices $\mathbf{H}$ and $\mathbf{G}$ and  average power 
			$P_t$ and the outputs are the parameters of the rotation model (i.e., 
			${\bm\lambda}^{\ast}$ and 
			${\bm\theta}^{\ast}$) as well as the corresponding secrecy rate $R^{\ast}$. As shown in Fig.~\ref{fig_obj_sys}, there are four blocks in the BFGS-based optimizer: GSVD for initialization, BFGS for optimization, eigenvalues rectifier, and  objective function \eqref{eq_optD3_1} evaluator.  In the following subsections, we introduce the functionality of each block and  the optimization algorithm.

			\begin{figure}[h]
				\centering
				\includegraphics[width=0.48\textwidth]{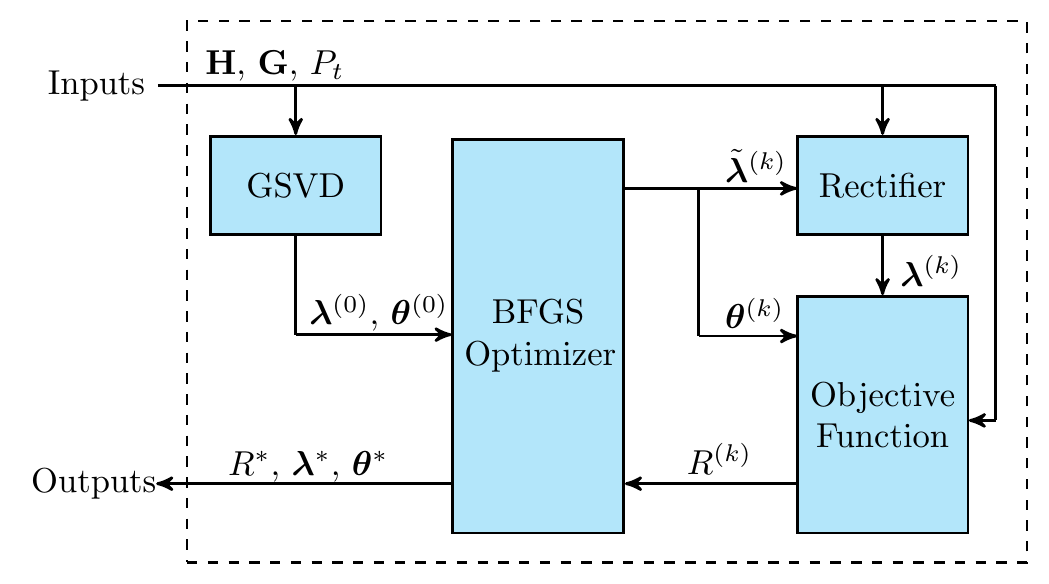}
				\caption{The system design of  rotation-BFGS method.}
				\label{fig_obj_sys}
			\end{figure}
			
			\subsection{Functionality of Each Block}
			\subsubsection{Initialization using GSVD}\label{sec_alg_sub_gsvd}
			While initial values of $\bm\lambda$ and $\bm\theta$  can be chosen 
			randomly, 	efficient initial values   can  save time by reducing the 
			number of iterations. For this reason, and knowing that GSVD-based 
			beamforming is a  good solution for this problem, we use GSVD to 
			find the  
			initial values ($\bm\lambda^{\rm(0)}$ and $\bm\theta^{\rm(0)}$)
			for our rotation-BFGS algorithm. 
			The solution given by GSVD-based beamforming provides a precoding matrix $ \mathbf{E} $ which 
			satisfies:
			\begin{subequations}
				\begin{align}
				\mathbf{H}\mathbf{E}=\mathbf{\Psi}_r\mathbf{C},\quad\label{eq_gsvda}\\
				\mathbf{G}\mathbf{E}=\mathbf{\Psi}_e\mathbf{D},\quad\label{eq_gsvdb}\\
				\mathbf{C}^T\mathbf{C}+\mathbf{D}^T\mathbf{D}=\mathbf{I},\label{eq_gsvdc}
				\end{align}	
			\end{subequations}
			where
			$\mathbf{E}\in\mathbb{R}^{n_t\times q} $,  $q=\min(n_t,n_r+n_e)$, 
			$\mathbf{C}^T\mathbf{C}={\rm diag}(c_i)$ and $\mathbf{D}^T\mathbf{D}={\rm 
				diag}(d_i)$, $i\in\{1,\hdots, q\}$, are 
			diagonal matrices, and  $\mathbf{\Psi}_r 
			\in\mathbb{R}^{n_r\times n_r}$ and $\mathbf{\Psi}_e
			\in\mathbb{R}^{n_e\times n_e} $ are orthonormal 
			matrices. Besides, the 
			power allocation matrix 
			$\mathbf{P}={\rm diag}(p_i)$ is determined by 
			\cite{fakoorian2012optimal}
			\begin{align}\label{eq_gsvdP}
			p_i=\begin{cases} 
			\max(0, \frac{2(c_i-d_i)/(\mu 
e_i)-2}{1+\sqrt{1-4c_id_i+4(c_i-d_i)c_id_i/(\mu 
					e_i)} }      ),	
			&\textmd{if } c_i>d_i,\\
			0, &\textmd{otherwise},
			\end{cases}
			\end{align}	
			in which $p_i$ and $e_i$ are the $i$th diagonal element of $ 
			\mathbf{P} 
			$ 
			and $ \mathbf{E}^T\mathbf{E} $ respectively, 
			and $\mu$ is the Lagrange multiplier to ensure
			\begin{align}\label{eq_gsvdLa}
			{\rm tr}(\mathbf{E}\mathbf{P}\mathbf{E}^T)=P_t.
			\end{align}
			GSVD-based  $\mathbf{Q}$ is then $ 
			\mathbf{E}\mathbf{P}\mathbf{E}^T $. Thus, we can determine  $\mathbf{V}^{(0)}$ 
			and $\mathbf{\Lambda}^{(0)}$ from eigenvalue decomposition of $\mathbf{E}\mathbf{P}\mathbf{E}^{T}$, i.e., from
			\begin{align}\label{eq_gsvd2Q}
			\mathbf{Q}^{(0)}=\mathbf{E}\mathbf{P}\mathbf{E}^{T}
			=\mathbf{V}^{(0)}\mathbf{\Lambda}^{(0)}\mathbf{V}^{(0)T}.
			\end{align} Then, ${\bm\lambda}^{(0)}= \rm diag (	\mathbf{\Lambda}^{(0)})$  and $ {\bm\theta}^{(0)} $ can be obtained from  $ \mathbf{V}^{(0)} $
			using  Algorithm~\ref{alg_solveAngle}. \\
			
			\subsubsection{BFGS Optimizer}\label{sec_alg_sub_bfgs}
			
			Different iterative optimization methods, such as gradient descent, 
			Newton's method,   and quasi-Newton algorithms may be used to 
			optimize  (P2). Among the above methods,
			Newton's method is the fastest method,  but it requires a local 
			Hessian matrix and its inverse, which may not exist or hard to obtain.  
			Quasi-Newton methods keep the 
			convergence advantage of Newton's method  without requiring the local 
			Hessian matrix and its inverse.
			An outstanding
			variety of quasi-Newton methods	is the BFGS algorithm \cite{luenberger1984linear}. We give a brief introduction to this method in the following.

			Given an unconstrained optimization problem
			\begin{align}\label{eq_bfgs_main_}
			\arg\min\limits_{\mathbf{x}} f(\mathbf{x}),  
			\end{align}
			with argument vector $\mathbf{x}$ and objective function $f(\mathbf{x})$,  
			BFGS algorithm updates the vector $\mathbf{x}$ iteratively according to 
			\begin{align}\label{eq_bfgs_x_iter_}
			\mathbf{x}^{(k+1)}=\mathbf{x}^{(k)} - \alpha^{(k)} \mathbf{M}^{(k)} \mathbf{g}^{(k)},\;k\geq0,
			\end{align} 
			in which  
			\begin{itemize}
				\item  $\alpha^{(k)}$ is the step size to minimize $f(\mathbf{x}^{(k+1)})$, which can be obtained via a line search.
				\item $\mathbf{g}^{(k)}$ is the gradient of $f(\cdot)$ at $\mathbf{x}^{(k)}$.
				\item  $\mathbf{M}^{(k)}$ is an approximation of 
				the inverse of the Hessian matrix.
			\end{itemize}
			The matrix $\mathbf{M}$ is initialized by a unity matrix, i.e.,  $\mathbf{M}^{(0)}=\mathbf{I}$, and is then updated as \cite{luenberger1984linear}  
			\begin{align}\label{eq_bfgs_iter_H_}
			\mathbf{M}^{(k+1)}=&\left(\mathbf{I}-\frac{   \bm{\delta}_\mathbf{x}^{(k)}\bm{\delta}_\mathbf{g}^{(k)T}     }
			{\bm{\delta}_\mathbf{g}^{(k)T}\bm{\delta}_\mathbf{x}^{(k)}}\right)
			\mathbf{M}^{(k)}
			\left(\mathbf{I}-\frac{\bm{\delta}_\mathbf{g}^{(k)}\bm{\delta}_\mathbf{x}^{(k)T}}
			{\bm{\delta}_\mathbf{g}^{(k)T}\bm{\delta}_\mathbf{x}^{(k)}}\right) \notag\\
			&+\frac{\bm{\delta}_\mathbf{x}^{(k)}\bm{\delta}_\mathbf{x}^{(k)T}}
			{\bm{\delta}_\mathbf{g}^{(k)T}\bm{\delta}_\mathbf{x}^{(k)}},
			\end{align}
			in which
			\begin{subequations}\begin{align}
				&\bm{\delta}_\mathbf{x}^{(k)}\triangleq\mathbf{x}^{(k+1)}-\mathbf{x}^{(k)},\label{eq_bfgs_3a_}\\
				&\bm{\delta}_\mathbf{g}^{(k)}\triangleq\mathbf{g}^{(k+1)}-\mathbf{g}^{(k)}\label{eq_bfgs_3b_},
				\end{align}\end{subequations}
			respectively,  represent the difference between the arguments and the gradients in two successive iterations. The gradient can be obtained analytically or numerically.

			Due to its fast convergence speed 
			and self-correcting property \cite{byrd1989tool}, BFGS is widely used in 
			unconstrained optimization problems. However, the problem in this 
			paper is a constrained optimization problem due to 
			\eqref{eq_optD3_2} and \eqref{eq_optD3_3}. To overcome this limitation,  we relax  the constraints of  (P2) when using the BFGS optimizer, but  we add a new block called a rectifier, as shown in  Fig.~\ref{fig_obj_sys}. The rectifying block is used to ensure 
			that the constraint on $\lambda_i$s are satisfied as elaborated on in what follows.  
			
			\subsubsection{Rectifying Eigenvalues}\label{sec_alg_sub_pen}
			To use BFGS,  the optimization problem should be unconstrained.
		However, the eigenvalues  in (P2)  need to be constrained. To overcome this issue, we rectify the eigenvalues  to ensure that the constraints in 
			\eqref{eq_optD3_2}-\eqref{eq_optD3_3} are satisfied. 
			
			Before talking about the rectification process, we highlight that we 
			only 
			optimize the first $n_t-1$ eigenvalues because the last 
			eigenvalue will be obtained by $\sum_{i=1}^{n_t}\lambda_i=P_t$.	
			That is,  we  use  {$P_t-\sum_{i=1}^{n_t-1}{\lambda}_i$} 
			as the value 
			of the last eigenvalue ($\lambda_{n_t}$) since in the MIMOME channel it is 
			known that optimal solution uses the total power 
			\cite{khisti2010secure,loyka2016optimal,vaezi2017journal}. 
			
			Suppose $\tilde{\bm \lambda}\in\mathbb{R}^{n_t-1}$ is the vector of 
			first $n_t-1$ unconstrained eigenvalues 
			obtained from the BFGS algorithm. 
			We obtain    $ \bm 
			\lambda\in\mathbb{R}^{n_t}$  (the rectified eigenvalue vector) from 
			\begin{align}\label{eq_rect_func}
			{\bm\lambda}={\rm r}(\tilde{\bm\lambda}, P_t), 
			\end{align}
			where the rectifying function $\rm r(\cdot,\cdot)$ is defined by the 
			following successive processes:
			\begin{subequations}
				\begin{align}
				&{\bm\lambda^+=[\tilde{\bm \lambda}]^+,} 
				\label{eq_lam_p1}\\
				&{\bar{\bm\lambda}}=
				\left\{
				\begin{array}{ll}
				{\bm\lambda^+}\cdot \frac{P_t}{\sum_{i=1}^{n_t-1}{\lambda^+_i}}, &\sum_{i=1}^{n_t-1}{\lambda^+_i} > P_t,\\
				{\bm\lambda^+}, &{\rm otherwise,}
				\end{array}
				\right.\label{eq_lam_p2}\\
				&{\bm\lambda}=\left[{\bar{\bm\lambda}},P_t-\sum_{i=1}^{n_t-1} 
				\bar{\bm\lambda}\right].  \label{eq_lam_p3}
				\end{align}
			\end{subequations}
			In \eqref{eq_lam_p1}, $[\cdot]^+$ is an element-wise operation 
			which forces all negative elements of $\tilde{\bm \lambda}$ to 0  
			while  keeping non-negative values unaltered.  This will take care of 
			the constraints in  \eqref{eq_optD3_3}. But, the elements of the new 
			vector ${\bm\lambda}^+$ may not still satisfy  \eqref{eq_optD3_2}.
			In such a case, in \eqref{eq_lam_p2}, we scale the new vector such that  the sum of the eigenvalues does not exceed $P_t$. Finally, in the \eqref{eq_lam_p3}, the last eigenvalue is added.
			Thus, the problem (P2) can be solved using an unconstrained optimization method, namely, the BFGS.
			
			This process illustrated  in Fig.~\ref{fig_paramP} for 
			$n_t=3$ and $P_t=5$, as an example. In this figure, the 
			blue  
			(shaded) area denotes the feasible region of the eigenvalues, red 
			points outside the region are unconstrained eigenvalues
			$\tilde {\bm\lambda} = [\lambda_1, \lambda_2, \cdots, 
			\lambda_{n_t-1}]$ (the output of the BFGS block), and green squares 
			denote their corresponding rectified eigenvalues. 
			With the proposed rectification in \eqref{eq_lam_p1}-\eqref{eq_lam_p2},  unrestricted inputs will be forced 
			to new points on the boundary of the feasible region. 
			For example, for the point $(-1,4)$ we have $\tilde {\bm\lambda}= 
			[-1,4]$, $\bar{\bm\lambda}=\bm\lambda^+= [0,4]$ (the second case 
			in \eqref{eq_lam_p2}) and  $\bm\lambda=[\lambda_1, \lambda_2, P - 
			\lambda_1- \lambda_2]= [0,4,1]$, following 
			\eqref{eq_lam_p1}-\eqref{eq_lam_p2}. Similarly, the red point 
			$(-1,-1)$ will result in  $\bm\lambda= [0,0,5]$. 
			
			\begin{figure}[h]
				\centering
				\includegraphics[width=0.4\textwidth]{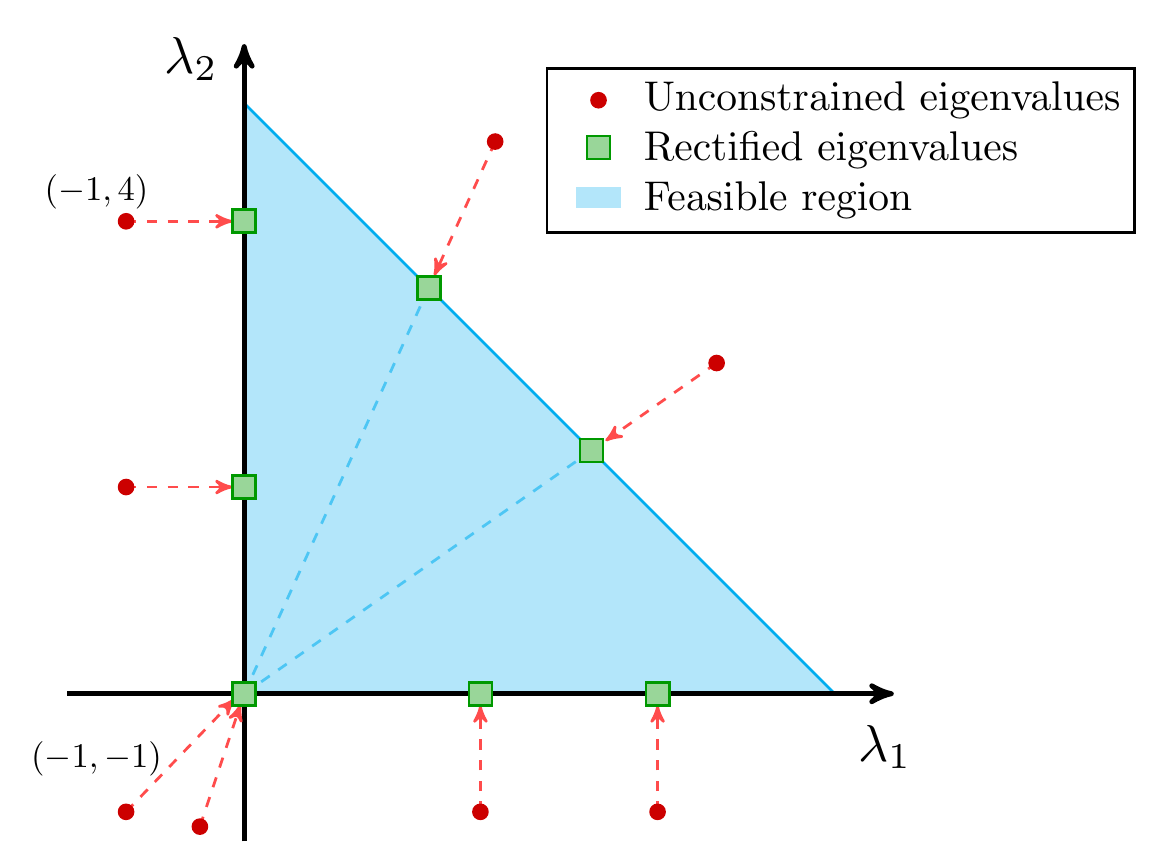}
				\caption{Illustration of eigenvalues rectification for $n_t=3$.}
				\label{fig_paramP}
			\end{figure}
			
			\subsubsection{The Objective Function}		
			In this paper, our goal is to maximize the secrecy rate $R$ defined by
			\begin{align}\label{eq_optDN_obj}
			&R(\bm{\lambda}, \bm{\theta})\triangleq
			\frac{1}{2}\log_2\frac{|\mathbf{I}_{n_t}+\mathbf{V}^T\mathbf{H}^T\mathbf{H}\mathbf{V}\mathbf{\Lambda}|}
			{|\mathbf{I}_{n_t}+\mathbf{V}^T\mathbf{G}^T\mathbf{G}\mathbf{V}\mathbf{\Lambda}|},
			\end{align}
			which is  a function of  $\bm\theta$ and  $\bm\lambda$ as shown 
			in   (P2). Noting that our problem is a maximization  rather than  a 
			minimization problem, in comparison to \eqref{eq_bfgs_main_},  we 
			define the  objective function $f(\mathbf{x}) =-R$.
			With this,  we can link BFGS and rotation method together. 
			
			\subsection{Rotation-BFGS Algorithm}\label{sec_alg_all}
			As shown in  Fig.~\ref{fig_obj_sys}, GSVD provides an initial value for the argument  $\mathbf{x}$ in the BFGS block. The rectifier  and the objective function can be considered together which require the current value of  $\mathbf{x}$  and provide corresponding achievable rate back to the BFGS. Then, BFGS will update $\mathbf{x}$ and check the termination condition. 
			
			To link the BFGS to the rotation model,  we need to define  the relation between $\mathbf{x}$, the arguments inside BFGS, and the parameters $\bm\lambda$ and $\bm\theta$.  This is given by 
			\begin{align}
			\mathbf{x}\triangleq[\tilde{\bm\lambda}, 
			{\bm\theta}],\label{eq_interF_x_}
			\end{align}
			where  $\tilde{\bm\lambda}$ is the first $n_t-1$ elements  of $\bm\lambda$ and  ${\bm\theta}$ is the same as we defined in \eqref{eq_theVec}. In this way, we reduce one argument (eigenvalue) for efficiency.
			Then,  in the $k$th iteration ($k\ge0$) of the rotation-BFGS algorithm,   the 
			relation between  argument $\mathbf{x}^{(k)}$ in the BFGS and the eigenvalues ${\bm\lambda}^{(k)}$ in precoding is
			\begin{subequations}
				\begin{align}
				&\mathbf{x}^{(k)}\triangleq[\tilde{\bm\lambda}^{(k)}, 
				{\bm\theta}^{(k)}],\label{eq_interF_x}\\
				&\bm{\lambda}^{(k)} \triangleq \rm r(\tilde{\bm\lambda}^{(k)}, P_t),
				\label{eq_interF_rla2la}
				\end{align}
			\end{subequations}
			in which $\tilde{\bm\lambda}^{(k)}$ denotes the  first $n_t-1$
			unconstrained eigenvalues and  $\bm\lambda^{(k)}$ is the  constrained (rectified)  eigenvalues which are obtained by \eqref{eq_rect_func}-\eqref{eq_lam_p3}. 
			Besides, the value of the function  is denoted as 
			\begin{align}
			&f^{(k)}\triangleq f(\mathbf{x}^{(k)})=-R({\rm r}({\tilde{\bm\lambda}}^{(k)},P_t),{\bm\theta}^{(k)}),\label{eq_interF_f}
			\end{align}
			in which ${\bm\lambda}^{(k)}$ and ${\bm\theta}^{(k)}$ can be obtained from 
			$\mathbf{x}^{(k)}$ according to \eqref{eq_interF_x}-\eqref{eq_interF_rla2la}, and $f(\cdot)$ denotes the mapping from $\mathbf{x}^{(k)}$ to $f^{(k)}$.
			Finally, the gradient vector $\mathbf{g}^{(k)}$ with respect to $\mathbf{x}^{(k)}$ is 
			obtained numerically. Specifically, the $i$th element of $\mathbf{g}$ in the $k$th iteration
			is given by
			\begin{align}
			{g}_i^{(k)}=\frac{f(\mathbf{x}^{(k)}+{\bm\epsilon}_i)-f^{(k)}}{|{\bm\epsilon}_i|},\label{eq_grad}
			\end{align} 
			where ${\bm\epsilon}_i$   has the same length as $\mathbf{x}$ and its all component are zero  except for the $i$th element which is a constant $\epsilon_1$.

			\begin{algorithm}[h]
				\caption{Rotation-BFGS Method}\label{alg_RotMwthod}
				\begin{algorithmic}[1]
					{
						\STATE	Initialize: $\epsilon_1=\epsilon_2=10^{-4}$, $k=0$, and 
						$\mathbf{M}^{(0)}=\mathbf{I}$;
						\STATE	Find $ {\mathbf{V}}^{(0)}$ and $\mathbf{\Lambda}^{(0)} $ using GSVD in \eqref{eq_gsvd2Q};
						\STATE	Find ${\bm{\lambda}}^{(0)}$ which is the diagonal of $ 
						\mathbf{\Lambda}^{(0)} $ ;
						\STATE	Find $\tilde{\bm{\lambda}}^{(0)}$ which is the first $(n_t-1)$ 
						elements in 			${\bm{\lambda}}^{(0)}$;
						\STATE	Find $\bm{\theta}^{(0)}$ using Algorithm~\ref{alg_solveAngle};
						\STATE Find $R^{(0)}$ using \eqref{eq_optDN_obj};
						\STATE Find $\mathbf{x}^{(0)}$, $\mathbf{g}^{(0)}$, and $f^{(0)}$ using 
						\eqref{eq_interF_x}-\eqref{eq_grad}; 
						\WHILE 1
						\STATE Find $\alpha^{(k)}=\alpha^{\ast}$  by line search using 
						Algorithm~\ref{alg_Lineserach};			
						\STATE Find $\mathbf{x}^{(k+1)}$ using \eqref{eq_bfgs_x_iter_};
						\STATE Find ${\bm\lambda}^{(k+1)}$ and ${\bm\theta}^{(k+1)}$ from 
						$\mathbf{x}^{(k)}$ using \eqref{eq_interF_x}-\eqref{eq_interF_rla2la};
						\STATE Find $R({\bm\lambda}^{(k+1)},{\bm\theta}^{(k+1)})$ using \eqref{eq_optDN_obj};
						\STATE Find  $\mathbf{g}^{(k+1)}$ and 
						$f^{(k+1)}$ using 
						\eqref{eq_interF_x}-\eqref{eq_grad};  
						\STATE Find  $\mathbf{M}^{(k+1)}$ using BFGS by \eqref{eq_bfgs_iter_H_}-\eqref{eq_bfgs_3b_};
						\IF {$|f^{(k+1)}-f^{(k)}|<\epsilon_2$}
						\STATE ${R}^{\ast}={-f}^{(k+1)}$, 
						${\bm\lambda}^{\ast}={\bm\lambda}^{(k+1)}$, and 
						${\bm\theta}^{\ast}={\bm\theta}^{(k+1)}$;
						\STATE Break;
						\ENDIF
						\STATE $k=k+1$;
						\ENDWHILE
						\STATE  Output: ${R}^{\ast}$, ${\bm\lambda}^{\ast}$, and ${\bm\theta}^{\ast}$.
					}
				\end{algorithmic}
			\end{algorithm}
			In our proposed rotation-BFGS method, as shown in Fig.~\ref{fig_obj_sys}, GSVD-based beamforming is  first applied  to find initial values of  ${\bm\lambda}^{(0)}$, ${\bm\theta}^{(0)}$ which are obtained by GSVD decomposition as described in 
			Section~\ref{sec_alg_sub_gsvd}. Thus
			$	\mathbf{x}^{(0)}= [\tilde{\bm\lambda}^{(0)}, {\bm\theta}^{(0)}]$. We 
			should also  
			highlight that $\tilde{\bm\lambda}^{(0)}$ contains the first $n_t-1$ elements of $\bm{\lambda}^{(0)}$.

			Algorithm~\ref{alg_RotMwthod} illustrates the details of the proposed 
			optimization process. Within this algorithm, we require a line 
			search. The  details of the line search are given in Algorithm~\ref{alg_Lineserach} in the Appendix. Algorithm~\ref{alg_RotMwthod} will 
			terminate 
			if the secrecy rate at two successive iterations are very close, i.e., when 
			their difference is smaller than a tolerance ($\epsilon_2$). 
			Algorithm~\ref{alg_RotMwthod}  guarantees a 
			global convergence based on the global convergence theorem 
			\cite[Chapter 7, pp. 196--204]{luenberger1984linear}, because the 
			optimized function \eqref{eq_interF_f} (the objective function 
			together 
			with the rectifier in Fig.~\ref{fig_obj_sys})  is continuous, and the 
			iteration in \eqref{eq_bfgs_x_iter_} can lead to a decreasing sequence, 
			i.e., $f(\mathbf{x}^{(k+1)}) \leq f(\mathbf{x}^{(k)})$.

			 	We analyze the complexity of the different algorithms here. The 
			 	computation of  matrix  multiplications and matrix inverse yields 
			 	the 
			 	complexity of $\mathcal{O}(L^3)$ where $L= \max(n_t, n_r, n_e)$. 
			 	GSVD-based precoding \cite{fakoorian2012optimal}, used  as the initial 
			 	point 
			 	generator, has the complexity of 
			 	$\mathcal{O}({L^3}+{L}\log(1/\epsilon))$ 
			 	\cite{park2015weighted} in which $\epsilon$ is the convergence 
			 	tolerance of the algorithm, while  bisection search requires 
			 	$\mathcal{O}(\log(1/\epsilon))$ iterations \cite{boyd2004convex}.   
			 	 Besides, the BFGS algorithm has the complexity 
			 	of $\mathcal{O}(n^2)$ \cite{nocedal2006numerical}, where $n$ is the size of input 
			 	variables which is the total number of optimized rotation parameters, i.e., $\frac{1}{2}n_t(n_t+1)$. 
			 	  Thus, the overall complexity of Algorithm~\ref{alg_RotMwthod} is 
			 	$\mathcal{O}({n_t^4+L^3+{L}\log(1/\epsilon)})$.
			 	AOWF \cite[Algorithm 1]{li2013transmit} yields $\mathcal{O}(\frac{L^3}{\epsilon}\log({1}/{ \epsilon}))$, in which $\mathcal{O}(\frac{1}{\epsilon})$ and 	$\mathcal{O}(\log(1/\epsilon))$ are the outer layer loop and inner bisection search, respectively.
			{It should also be mentioned that rotation-Fmin has the same 
			asymptotic complexity as 
				rotation-BFGS. 
				However, for small values of $n_t$, rotation-BFGS has about an order of magnitude smaller complexity than  rotation-Fmin, as we show in Section~\ref{sec_V_result}. }

		}


		
		\section{Numerical Result}\label{sec_V_result} 
			\begin{table}[t]
		\caption{Achievable rate (in bps/Hz) of each method 
			for 
			$n_t=3$ and $P_t=30W$.}
		\label{tab_rate_nt3}
		\centering
		{
			\begin{tabular}{|c|c|cccccc|} 
				
						\hline
				\multicolumn{8}{|c|}{Rotation-BFGS}                                            
				\\ \hline
				\multicolumn{2}{|c|}{\multirow{2}{*}{$n_t=3$}} & 
				\multicolumn{6}{c|}{$n_e$}                          \\ \cline{3-8}
				\multicolumn{2}{|c|}{}                         & 1      & 2      & 3      & 
				4      & 5      & 6      \\ \hline
				\multirow{6}{*}{$n_r$}           
				&1&2.58&    1.99&    1.14&    0.63&    0.39&    0.23\\
				&2&3.91&    2.99&    1.75&    1.09&    0.71&    0.43\\
				&3&4.84&    3.59&    2.24&    1.47&    1.01&    0.70\\
				&4&5.48&    4.16&    2.70&    1.73&    1.24&    0.89\\
				&5&5.97&    4.52&    3.01&    2.03&    1.45&    1.09\\
				&6&6.46&    4.88&    3.30&    2.35&    1.76&    1.28\\ \hline
				
				\hline\hline
				\multicolumn{8}{|c|}{{Rotation-Fmin}}                             
				\\ \hline
				
				\multicolumn{2}{|c|}{\multirow{2}{*}{$n_t=3$}} & 
				\multicolumn{6}{c|}{$n_e$}                          \\ 
				\cline{3-8}
				\multicolumn{2}{|c|}{}                         &{1}      
				&{2      }&{ 3      }&{
				4      }&{5      }&{ 6      }\\ \hline
				\multirow{6}{*}{$n_r$}           
				&{1}&{2.58}&{
				1.99}&{ 1.14}&{0.63}&{
				0.39}&{ 0.23}\\
				&{2}&{3.91}&{
				2.99}&{1.75}&{  1.09}&{  
				0.71}&{ 0.43}\\
				&{3}&{4.84}&{
				3.59}&{2.24}&{ 1.47}&{
				1.01}&{0.70}\\
				&{4}&{5.48}&{
				4.16}&{2.70}&{1.73}&{ 
				1.24}&{   0.89}\\
				&{5}&{5.97}&{
				4.52}&{ 3.01}&{  2.03}&{
				1.45}&{ 1.09}\\
				&{6}&{6.46}&{
				4.88}&{  3.30}&{2.35}&{
				1.76}&{   1.28}\\ \hline

				\hline\hline
				\multicolumn{8}{|c|}{GSVD \cite{fakoorian2012optimal}}                                                          
				\\ \hline
				\multicolumn{2}{|c|}{\multirow{2}{*}{$n_t=3$}} & 
				\multicolumn{6}{c|}{$n_e$}                    \\ \cline{3-8}
				\multicolumn{2}{|c|}{}                         & 1     & 2     & 3     & 
				4     & 5     & 6     \\ \hline
				\multirow{6}{*}{$n_r$}          
				&1&2.56&    1.82&    1.12&    0.63&    0.39&    0.23\\
				&2&2.97&    2.87 &   1.74 &   1.09&    0.70 &   0.43\\
				&3&4.41 &   3.50  &  2.23  &  1.46 &   1.01  &  0.70\\
				&4&5.17  &  4.10   & 2.69   & 1.73  &  1.24   & 0.89\\
				&5&5.76   & 4.46    &3.00    &2.02   & 1.45    &1.09\\
				&6&6.28    &4.83    &3.29    &2.35    &1.76    &1.28\\ \hline
				
				\hline\hline
				\multicolumn{8}{|c|}{AOWF \cite{li2013transmit}}                                                        
				\\ \hline
				\multicolumn{2}{|c|}{\multirow{2}{*}{$n_t=3$}} & 
				\multicolumn{6}{c|}{$n_e$}                                 \\ \cline{3-8}
				\multicolumn{2}{|c|}{}                         & 1        & 2        & 3      
				& 4       & 5       & 6       \\ \hline
				\multirow{6}{*}{$n_r$}          
				&1&2.58&    1.99 &   1.13&    0.63&    0.38&    0.23\\
				&2&3.92&    2.98 &   1.74&    1.08 &   0.69&    0.42\\
				&3&4.86&    3.58 &   2.23&    1.45 &   0.99&    0.68\\
				&4&5.49&    4.15 &   2.69&    1.70 &   1.20&    0.85\\
				&5&5.99&    4.51 &   2.98&    1.97 &   1.39&    1.02\\
				&6&6.47&    4.87 &   3.26&    2.28 &   1.67&    1.18\\ \hline
			\end{tabular}
		}
	\end{table}
	
	\begin{table}[t]
		\caption{Achievable rate (in bps/Hz) of each method 
			for 
			$n_t=4$ and $P_t=30W$.}
		\label{tab_rate_nt4p}
		\centering
		{
			\begin{tabular}{|c|c|cccccc|}
				\hline
			\multicolumn{8}{|c|}{Rotation-BFGS}                                            
			\\ \hline
			\multicolumn{2}{|c|}{\multirow{2}{*}{$n_t=4$}} & 
			\multicolumn{6}{c|}{$n_e$}                          \\ \cline{3-8}
			\multicolumn{2}{|c|}{}                         & 1      & 2      & 3      & 
			4      & 5      & 6      \\ \hline
			\multirow{6}{*}{$n_r$}  
			
			&1&3.04&    2.63&    2.06&    1.25&    0.77&    0.44\\
			&2&4.72&    4.01&    3.12&    2.02&    1.32&    0.87\\
			&3&5.91&    4.97&    3.78&    2.60&    1.72&    1.21\\
			&4&6.81&    5.68&    4.40&    3.05&    2.08&    1.55\\
			&5&7.63&    6.28&    4.82&    3.46&    2.50&    1.83\\
			&6&8.24&    6.76&    5.22&    3.78&    2.74&    2.10\\ \hline
				
				\hline\hline
				
				\multicolumn{8}{|c|}{Rotation-Fmin}

				\\ \hline
				\multicolumn{2}{|c|}{\multirow{2}{*}{$n_t=4$}} & 
				\multicolumn{6}{c|}{$n_e$}                          \\ 
				\cline{3-8}
				\multicolumn{2}{|c|}{}                         &{ 1      
				}&{2      }&{3      }&{
				4      }&{ 5      }&{6      }\\ \hline 
				\multirow{6}{*}{$n_r$}  
				
				&{1}&{3.04}&{
				2.63}&{   2.06}&{   1.25}&{
				0.77}&{  0.44}\\
				&{2}&{4.73}&{ 
				4.02}&{   3.12}&{2.02}&{ 
				1.32}&{ 0.87}\\
				&{3}&{5.91}&{
				4.98}&{3.78}&{2.60}&{
				1.72}&{ 1.21}\\
				&{4}&{6.81}&{
				5.69}&{4.40}&{ 3.05}&{  
				2.08}&{1.55}\\
				&{5}&{7.63}&{
				6.29}&{4.82}&{ 3.46}&{
				2.50}&{1.83}\\
				&{6}&{8.24}&{
				6.77}&{ 5.22}&{ 3.78}&{
				2.74}&{ 2.10}\\ \hline

				\hline\hline
				\multicolumn{8}{|c|}{GSVD \cite{fakoorian2012optimal}}                                                          
				\\ \hline
				\multicolumn{2}{|c|}{\multirow{2}{*}{$n_t=4$}} & 
				\multicolumn{6}{c|}{$n_e$}                    \\ \cline{3-8}
				\multicolumn{2}{|c|}{}                         & 1     & 2     & 3     & 
				4     & 5     & 6     \\ \hline
				\multirow{6}{*}{$n_r$}          
				&1&3.03&    2.59&    1.82&    1.22&    0.76&    0.44\\
				&2&4.27 &   3.02 &   2.97 &   1.99 &   1.31 &   0.87\\
				&3&3.87  &  4.42  &  3.67  &  2.57  &  1.72  &  1.21\\
				&4&5.80   & 5.33   & 4.31   & 3.03   & 2.08   & 1.55\\
				&5&6.91    &6.03   & 4.75    &3.44    &2.50    &1.82\\
				&6&7.70    &6.56    &5.16    &3.76    &2.73    &2.09\\ \hline
				
				\hline\hline
				\multicolumn{8}{|c|}{AOWF \cite{li2013transmit}}                                                        
				\\ \hline
				\multicolumn{2}{|c|}{\multirow{2}{*}{$n_t=4$}} & 
				\multicolumn{6}{c|}{$n_e$}                                 \\ \cline{3-8}
				\multicolumn{2}{|c|}{}                         & 1        & 2        & 3      
				& 4       & 5       & 6       \\ \hline
				\multirow{6}{*}{$n_r$}          
				&1&3.04&    2.63&    2.06&    1.24&    0.76&    0.44\\
				&2&4.74 &   4.04 &   3.12 &   2.01 &   1.31 &   0.86\\
				&3&5.94  &  4.98  &  3.77  &  2.58  &  1.70  &  1.18\\
				&4&6.82   & 5.68   & 4.39   & 3.04   & 2.06   & 1.52\\
				&5&7.63    &6.29    &4.81    &3.44    &2.46    &1.78\\
				&6&8.24    &6.75    &5.21    &3.74    &2.69    &2.03\\ \hline
			\end{tabular}
		}
	\end{table}
		In this section, extensive numerical results are provided to illustrate the 
		performance of the proposed rotation-BFGS method. {Four} 
		methods are taken into 
		account: 
		
		
		\begin{itemize}
			\item \textbf{Rotation-BFGS}: the proposed rotation-BFGS parameterization solved by the BFGS method.
			\item \textbf{Rotation-Fmin}: the proposed rotation 
			modeling 
			solved using MATLAB convex optimization tool {\tt fmincon}. 
			\item \textbf{GSVD}: GSVD-based beamforming with optimal power allocation \cite{fakoorian2012optimal},  as described in 
			Section~\ref{sec_alg_sub_gsvd}. 
			\item \textbf{AOWF}: alternating optimization and water-filling\cite{li2013transmit}.
		\end{itemize} 
		All results are based on averaging over 1000 realizations of independent  $\mathbf{H} $ and $\mathbf{G} $. 
		These entries of $ \mathbf{H} $ and $ \mathbf{G}$ are generated based 
		on the standard Gaussian distribution, i.e., $ \mathcal{N} (0, 1)$. Two 
		performance metrics are considered. In the first part, we focus on the 
		achievable secrecy rate. In the second part, the time consumption is 
		analyzed for different methods in various antenna configurations.

		\begin{figure}[t]
			\centering
			\subfigure[{Relative rate improvement} ($\%$) of the 
			proposed method  
			 to  
			GSVD.]{
				\includegraphics[width=0.45\textwidth]{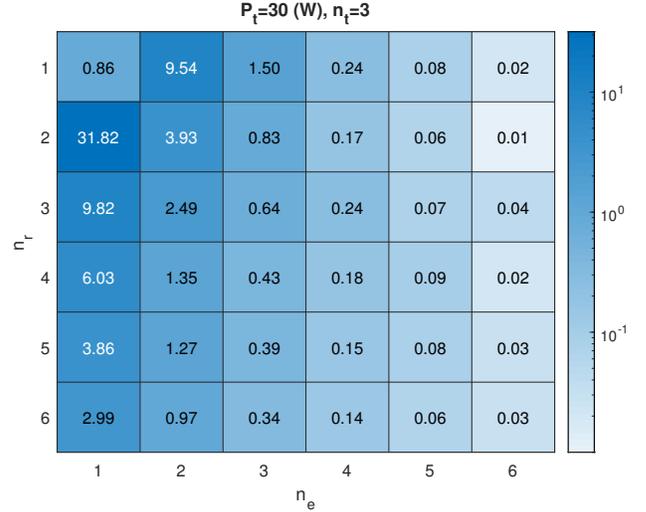}
				\label{fig_WTcomp_a}}
			\subfigure[{Relative rate improvement} ($\%$)  of the 
			proposed method 
		to AOWF.]{ 
				\includegraphics[width=0.45\textwidth]{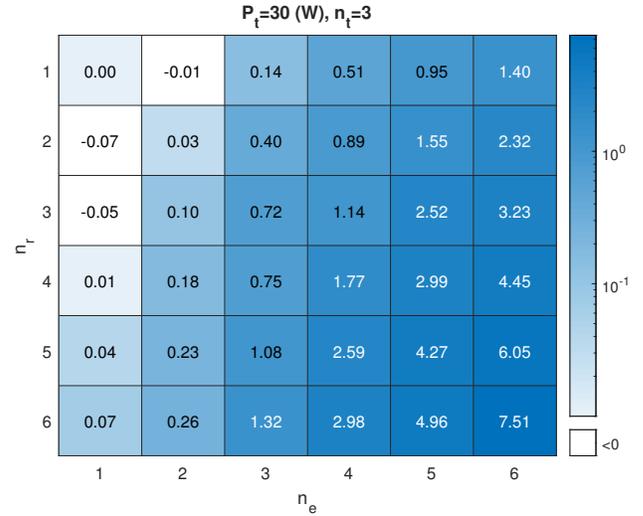}\label{fig_WTcomp_b}}
			\caption{Comparisons between the secrecy transmission 
			rates of 
				the proposed  method with GSVD and AOWF.}  
			\label{fig_WTcomp}
		\end{figure}

		 \subsection{Achievable Secrecy Rate }\label{sec_res_sub1}

		 We evaluate the performance of the Rotation-BFGS, Rotation-Fmin, 
		 GSVD, 
		and AOWF approaches with respect to the  variation of  the 
		number of antennas. The average secrecy rates are listed 
		 in Tables~\ref{tab_rate_nt3} and \ref{tab_rate_nt4p}  respectively for $n_t=3$ 
		 and 
		$n_t=4$. As expected, increasing the number of antennas at the eavesdropper decreases secrecy rate. On the contrary, 
		the secrecy rate will increase when the legitimate receiver or transmitter 
		has higher number of antennas.  To better appreciate the 
		improvement 
		due to	our proposed numerical method,
		relative secrecy rate  improvement between Rotation-BFGS and GSVD 
		and Rotation-BFGS and AOWF is investigated in the following and 
		illustrated in Fig.~\ref{fig_WTcomp}. Let us define relative 
		rate improvement with respect to GSVD and AOWF, respectively,  as
			\begin{subequations}
				\begin{align}
				\eta_g = \frac{R_r-R_g}{R_g}\times 100\%,\\
				\eta_a = \frac{R_r-R_a}{R_a}\times 100\%,
				\end{align}
			\end{subequations}
			where $R_r$, $R_g$, and $R_a$ represent the secrecy rate achieved by 
		the proposed method, GSVD, and AOWF, respectively\footnote{ 
		Rotation-Fmin is similar to rotation-BFGS in terms of the 
		achievable rate (details are in Table~\ref{tab_rate_nt3} and 
		\ref{tab_rate_nt4p}), so relative rate improvements of that to GSVD and 
		AOWF 
		are the same as those in rotation-BFGS.}. As 
		shown in 
		Fig.~\ref{fig_WTcomp},  the proposed approach is capable of 
		achieving the same or  better secrecy rate in almost any antenna setting. The darker the cell color, the higher is improvement achieved by the proposed method 
		
			There is a clear pattern that 
			when the eavesdropper has 
		 a smaller number of antennas than the transmitter, the proposed method outperforms GSVD. The best case is  
		  $[n_r, n_e]=[2,1]$ where the improvement is about 32\%. 
		  On the other hand, for cases with larger $n_e$,  the rotation-BFGS 
		performs 
		better than AOWF. 
		 The best case is  
		$[n_r, n_e]=[6,6]$ where the improvement is 7.5\%. 	
		This pattern not only exist when $n_t=3$ but also at higher $n_t$s,  at least $n_t=4$. Some more detailed comparisons are presented as follows:}
		
	 \begin{figure}[tbp]
		\centering
		\subfigure[Achievable  secrecy rate  with relatively  small
		$n_e$s.]{
			\includegraphics[width=0.460\textwidth]{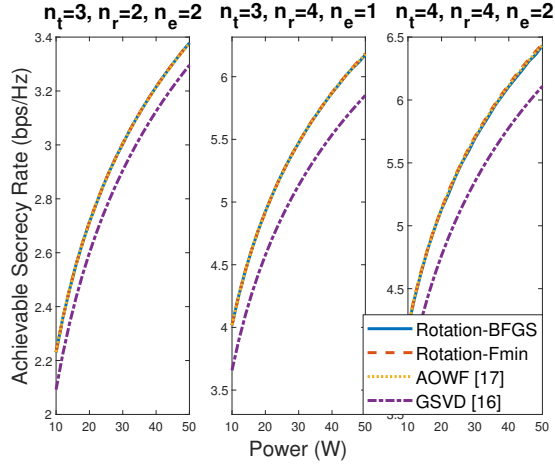}
			\label{fig_finallcomp02}}
		\subfigure[Achievable  secrecy rate  with relatively large 
		$n_e$s.]{
			\includegraphics[width=0.460\textwidth]{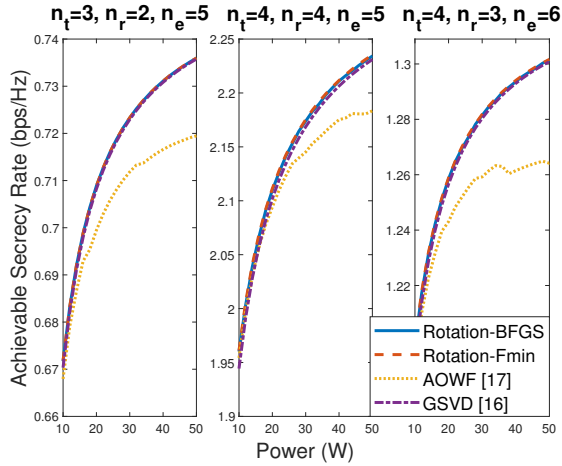}
			\label{fig_finallcomp03}}
		\caption{Secrecy rate of the MIMOME channel versus the 
		transmit 
			power. The proposed algorithms  (Rotation-BFGS and Rotation-Fmin) 
			are 
			compared 
			with GSVD and AOWF.}
		\label{fig_finallcomp}
	\end{figure}

		\subsubsection{The MIMOME with small 
		$n_e$}\label{sec_res_subsub1}
		GSVD-based beamforming fails to get close to the secrecy capacity. This 
		phenomenon has been verified by previous literature 
		\cite{vaezi2017journal}. 	As can be seen in Fig.~\ref{fig_finallcomp02},  
		rotation-BFGS and AOWF can achieve similar results. On the contrary,
		GSVD clearly has a gap with those methods which results from  
		sub-optimality of GSVD. 


		\subsubsection{The MIMOME with large 
		$n_e$}\label{sec_res_subsub2}
		In this setting,  AOWF does not perform very well because the Lagrange multiplier of AOWF cannot be obtained properly. 
		As illustrated in Fig.~\ref{fig_finallcomp03} and Fig.~\ref{fig_WTcomp}, our proposed method outperforms AOWF in this regime. 
		
%
		From the simulation results, it is seen that the proposed rotation-BFGS is robust in a wide range of practical antennas settings on each node. This is a big advantage as the 
		robustness towards the variation of $n_e$ is necessary to guarantee secure  
		communication. We note that the eavesdropper can have any number of 
		antennas, and the solution for wiretap channels should be robust to such variations. 	
	 This makes the proposed precoding highly competitive the existing solutions of the MIMOME channel. 
					\begin{figure}[t]
		\centering
		\subfigure[$n_t=3$.]{
			\includegraphics[width=0.45\textwidth]{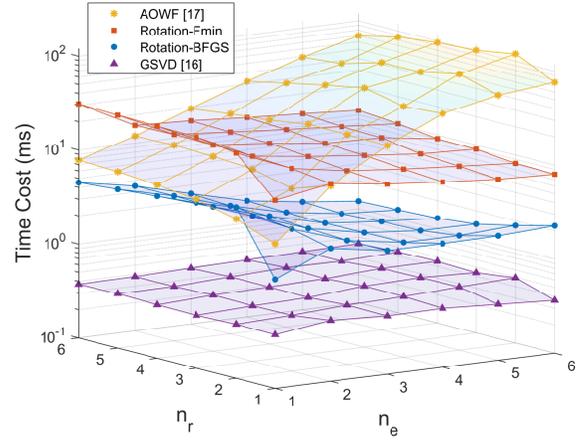}
			\label{fig_time_nt3}}
		\subfigure[$n_t=4$.]{
			\includegraphics[width=0.45\textwidth]{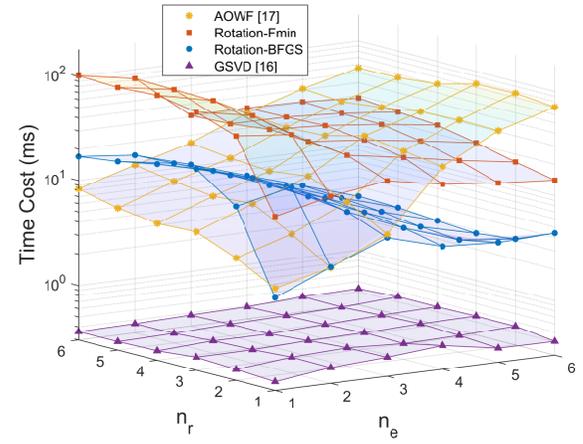}
			\label{fig_time_nt4}}
		\caption{Time costs of the proposed methods, GSVD 
			and AOWF, for (a) $n_t=3$, (b) $n_t=4$, with  $n_r$ and $n_e$ 
			from  1 to 6.}
		\label{fig_time} 
	\end{figure}

		\subsection{Time Consumption}\label{sec_res_sub2}
		Besides computational complexity, the execution time of each 
		algorithm is also important and provides a means to evaluate the 
		complexity. This 
		evaluation is shown in 
		Figs.~\ref{fig_time_nt3}-\ref{fig_time_nt4} for  $n_t=3$ 
		and 
		$n_t=4$ 
		respectively, which  are  
		averaged over 1000 channel realizations in each setting.  GSVD has, by 
		far, the best time cost since it  transforms the problem into finding the 
		best Lagrange multiplier. However, as we illustrated 
		previously, the precoding provided by GSVD may be far from the 
		optimal solution. The proposed rotation-BFGS 
		method clearly 
		outperforms AOWF, whose time consumption is very high, 
		especially when $n_e\in[4,5,6]$. The gap becomes
		larger when $n_e$ increases. In such cases, AOWF does not perform well either in 
		terms of achievable rate or time cost. On the other hand, our approach is robust in any antenna setting and more efficient in time 
		cost compared with AOWF. Rotation-Fmin achieves almost 
		the same performance, but its execution time is higher. This is 
		because 
		\texttt{fmincon} includes BFGS and interior-point method, {and instead 
		of the interior-point method, we 
			use a rectifier which is more efficient in terms of programming}.

			%
			In summary, the proposed rotation-BFGS method gives robust precoding and power allocation  
			for the MIMOME channel. By applying BFGS as an optimizer, the 
			rotation method has a better performance compared to other numerical 
			methods like AOWF in secrecy and time complexity.  In addition, the framework we provided in
			this paper can be used to solve many other problems, some listed in Section~\ref{sec_intro_sub_rWork}.

			\section{Conclusions}\label{sec_conclu}
			In this paper, we have developed a rotation-based method, called rotation-BFGS,  for precoding and 
			power allocation of Gaussian MIMOME channels. In this method,
			the transmit covariance matrix is constructed using Givens rotation matrices.  With this construction,  the PSD constraint of the transmit covariance matrix is removed and  the capacity   optimization problem is simplified.  The precoding (rotation) matrix and power allocation coefficients have been obtained iteratively using a modified BFGS algorithm. Compared to existing  approaches, the proposed method is robust  and performs well  independent of the number of  antennas at each node. The proposed 
			method  outperforms the GSVD-based beamforming when 
			$n_e< n_t$  and  AOWF particularly when $n_e\ge n_t$.
			This approach can also use the results of existing
			precoding and power allocation methods, such as the GSVD-based approach, as an initial point to expedite finding the  solution.

			In addition, the proposed rotation-BFGS method  has a great potential for finding precoding and 
			power allocation in various other applications, including in MIMO 
			broadcast channel and MIMO channel with and energy harvesting constraints, both with and without secrecy. Future works will focus on  further improving the efficiency of
			solving parameters and extension of this approach to other related problems.




			\appendices

\section{Line Search Algorithm for the BFGS Method}\label{app_bfgs}
The line search method we use is the Golden section search  
\cite{luenberger1984linear} and summarized in 
Algorithm~\ref{alg_Lineserach}, where $\alpha_i$ and  $f_i$ for 
$1\leq 
	i\leq4$ are temporary records for line search steps size and function 
values.

\begin{algorithm}[h]
	\caption{Line Search (Golden Section 
		\cite{luenberger1984linear}) for the BFGS }\label{alg_Lineserach}
	{
		\begin{algorithmic}[1]
			
			\STATE Requires from Algorithm~\ref{alg_RotMwthod}:  
			$\epsilon_2$, 
			$\mathbf{x}^{(k)}$, 
			$\mathbf{M}^{(k)}$, $\mathbf{g}^{(k)}$, 
			and $f^{(k)}$;
			\STATE Define: $\mathbf{x}\triangleq\mathbf{x}^{(k)}$, 
			$\mathbf{d}\triangleq\mathbf{M}^{(k)}\mathbf{g}^{(k)}$;
			\STATE	Initialize: $\epsilon_3=5\times 10^{-4}$, $\tau_1=3$, 
			$\tau_2=0.382$, 
			and $\tau_3=0.618$;
			\STATE Initialize: ${\bm \alpha}\triangleq[\alpha_1, 
			\alpha_2,\alpha_3,\alpha_4]=[0,0,0,0.1]$;
			\STATE Initialize: ${\bm 
				f}\triangleq[f_1,f_2,f_3,f_4]=[f^{(k)},0,0,f(\mathbf{x}-
			\alpha_4\mathbf{d})]$;
			\WHILE  {$\alpha_4<20$ and $f_1<f_4$}
			\STATE  Let $\alpha_4=\tau_1\alpha_4$ and 
			$f_4=f(\mathbf{x}-\alpha_4\mathbf{d})$;
			\ENDWHILE
			\STATE Let $\alpha_2=\tau_2\alpha_4$ and 
			$\alpha_3=\tau_3\alpha_4$;
			\STATE Let $f_2=f(\mathbf{x}-\alpha_2\mathbf{d})$ and 
			$f_3=f(\mathbf{x}-\alpha_3\mathbf{d})$;
			\WHILE {$\alpha_4-\alpha_1>\epsilon_3$ and 
				$\max({\bm f})-\min({\bm f})>\epsilon_2$}
			\STATE Define: $m$  is the index of the smallest element in ${\bm 
				f}$; 
			\STATE Define: $q_1\triangleq\max(1,m-1)$ and 
			$q_2\triangleq\min(4,m+1)$;
			\STATE Let $\alpha_1 = \alpha_{q_1}$ and $\alpha_4 = 
			\alpha_{q_2}$;
			\STATE Let $\alpha_2 =  \alpha_1+\tau_2(\alpha_4-\alpha_1)$ and 
			$\alpha_3 =  \alpha_1+\tau_3(\alpha_4-\alpha_1)$; 
			\IF {$m=2$}
			\STATE  Let $f_3=f_2$ and $f_2=f(\mathbf{x}-\alpha_2\mathbf{d})$;
			\ELSIF {$m=3$}
			\STATE Let  $f_2=f_3$ and  $f_3=f(\mathbf{x}-\alpha_3\mathbf{d})$;
			\ELSE
			\STATE Let $f_2=f(\mathbf{x}-\alpha_2\mathbf{d})$ and 
			$f_3=f(\mathbf{x}-\alpha_3\mathbf{d})$;
			\ENDIF
			\ENDWHILE
			\STATE Output: $\alpha^{\ast}=\alpha_m$.
	\end{algorithmic}}
\end{algorithm}



			
%
%
%
%
%
%
%
%
%
			
			
			\typeout{}
			\balance
			\bibliography{REF_commu_v1.0}
			\bibliographystyle{ieeetr}

\end{document}